\documentclass[conference]{IEEEtran}
\usepackage[top=2cm, bottom=2cm, left=2cm, right=2cm]{geometry}
\usepackage{algorithm}
\usepackage{algorithmicx}
\usepackage{algpseudocode}
\usepackage{amsmath}

\usepackage{diagbox}
\usepackage{latexsym}

\usepackage{subfigure}

\usepackage{caption}
\usepackage{graphicx, subfig}

\usepackage{tabularx}
\usepackage{rotating}
\usepackage{booktabs}

\ifCLASSINFOpdf
\else
\fi
\usepackage{url}

\begin{document}
%

\title{
A DNS Tunnel Sliding Window Differential Detection Method Based on Normal Distribution Reasonable Range Filtering
}

\author{
\IEEEauthorblockN{$Xin\ Ma^{1}$, $Shize\ Guo$, $Zhisong\ Pan^{1}$, $Bin\ Liu$, $Kaolin\ Jiang^{1}$, $Ming\ Chen$, $Shijiao\ Tang^1$}
\IEEEauthorblockA{$^{1}$Command and Control Engineering College, Army Engineering University of PLA, Nanjing 210007, China\\
xin.ma0206@gmail.com, szguo.uestc@gmail.com, hotpzs@hotmail.com,\\
 benliuchina@hotmail.com, jiangkolin@foxmail.com, chenming@163.com, shijiaotang2020@163.com
}


%


\IEEEoverridecommandlockouts
\makeatletter\def\@IEEEpubidpullup{9\baselineskip}\makeatother
\hspace{\columnsep}\makebox[\columnwidth]{}}

\maketitle

\begin{abstract}

A covert attack method often used by APT organizations is the DNS tunnel, which is used to pass information by constructing C2 networks. 
And they often use the method of frequently changing domain names and server IP addresses to evade monitoring, which makes it extremely difficult to detect them.
However, they carry DNS tunnel information traffic in normal DNS communication, which inevitably brings anomalies in some statistical characteristics of DNS traffic, so that it would provide security personnel with the opportunity to find them.
Based on the above considerations, this paper studies the statistical discovery methodology of typical DNS tunnel high-frequency query behavior. Firstly, we analyze the distribution of the DNS domain name length and times and finds that the DNS domain name length and times follow the normal distribution law.
Secondly, based on this distribution law, we propose a method for detecting and discovering high-frequency DNS query behaviors of non-single domain names based on the statistical rules of domain name length and frequency and we also give three theorems as theoretical support.
Thirdly, we design a sliding window difference scheme based on the above method.
Experimental results show that our method has a higher detection rate.
At the same time, since our method does not need to construct a data set, it has better practicability in detecting unknown DNS tunnels.
This also shows that our detection method based on mathematical models can effectively avoid the dilemma for machine learning methods that must have useful training data sets, and has strong practical significance.

\end{abstract}



%

\textbf{ Keywords:} DNS tunnel, normal distribution, abnormal points, reasonable range

\section{Introduction}


DNS tunnels \cite{DNS} have become one of the threats to network security. They are usually used by malware, new Trojan horses \cite{luo2017} and APT organizations, etc to transfer information or issue remote control instructions. In real-life scenarios, people usually deploy security devices at the boundaries of corporate networks to detect DNS tunnels.

There are some research about the detection of DNS tunnels, which are mainly relied on machine learning or statistical laws.
For example, Ahmed \cite{almusawi2018dns} uses SVM model to detect DNS tunnel. Born \cite{Born2010Detecting} \cite{Born2010NgViz} and Homem \cite{Homem2018Information} use information entropy to detect DNS tunnel by the confusion degree of domain name letters.
But, to our best knowledge, these existing methods, whether based on machine learning or statistical rules, usually need a labeled database to train models or adjust parameters. However, the DNS tunnel data used in their methods are usually constructed by open source software. As Yassine summarizes DNS tunnel detection methods in the work of \cite{Yassine2018}, people usually use different data sets, because there are few public DNS data sets. We will also introduce them in the next section.

The purpose of malicious authors using DNS tunnels is to pass information, this inevitably generates information transmission bandwidth, that is, a certain amount of DNS tunnel messages are always required to pass information, except very special cases, and only in which, no information or very little information be output. Therefore, in normal cases, a large number of DNS packets contain a certain amount of DNS tunnel packets, which will inevitably cause statistical anomalies.
Further considering the information transmission bandwidth, this requires a certain period and a large amount of data, which is suitable for deployment at the network exit of large enterprises. At the same time, there will also be high-frequency visits to some popular websites in these large enterprise traffic.


At the same time, since we are studying the network exit traffic of large enterprises, it will generate a lot of DNS access every day. However, DNS tunnel data is often constructed using known tunnel software, which is far from DNS tunnel data in real cyber threats. Faced with so many unknown DNS traffic, how to solve the DNS tunnel detection problem in the real environment requires us to propose a new solution.

Based on the above considerations, we propose our detection scheme.
In this paper, we mainly focus on the analysis of the general law of DNS traffic in the large enterprise network. We find that the domain name length and frequency obey the Normal Distribution. Based on this, we put forward some hypothesis, and use the theory to prove the hypothesis by drawing on the idea of the KS test \cite{ks}. 
Then, we use these theories to count out the reasonable range of fitting curve for filtering the abnormal point.
By doing so, we propose a theoretical detection model, and construct a sliding window differential detection framework.



The contributions of this paper are mainly on the following aspects:

Firstly, we find that the length and frequency of domain names roughly obey the normal distribution law approximately. A reasonable explanation is that people tend to use domain names of a certain length range. Short and long domain names are both rare, which will make a high in the middle of the curve and a low at both ends.


Secondly, we propose a theoretical detection model based on 5 hypotheses and 3 theorems. The model build a reasonable range of fitting curve for filtering burr points, that is abnormal point beyond the reasonable range.


Thirdly, based on the theoretical detection model, We propose a sliding window differential detection framework that make difference set on two adjacent windows for sudden appearing DNS.
It further improves detection efficiency and reduces false positive rate.


By testing, our framework can effectively detect DNS tunnel traffic in the actual complex network environment and achieve good results.



\section{Related work}

We are targeting the large enterprises DNS traffic from national project, which need us to provide DNS tunnel detection in its network boundary. This is different from the work of most researchers, and the biggest difference is the source of the data. Our data is a large amount of real network data, and it is unknown.
For this reason, we also comb and analyze some existing research methods on common DNS tunnels, to provide a reference. Next, we will give a brief introduction from three aspects: data source, data processing and detection model.

\subsubsection{Data Source}
Because it is difficult to mark DNS tunnel traffic after collecting actual network DNS traffic, there are few open DNS traffic data sets  until now. Meanwhile, DNS tunnel technology continues to update and develop, so that potential undetected DNS tunnel traffic is difficult to identify. 

In related research, the DNS traffic dataset is built by mixing normal DNS traffic captured and DNS tunnel traffic generated by open sources DNS tunnels software, such as iodine, dns2tcp, and tcp2dns.
Zhang \cite{zhang2013} and Berg \cite{Berg2019} use open source tools and campus traffic to construct DNS traffic data sets. 
Luo \cite{luo2017} uses their own DNS tunnel tool to construct DNS tunnel traffic and uses this traffic to determine the effectiveness of the algorithm. 
Ahmed \cite{almusawi2018dns} uses iodine to construct DNS tunnel traffic based on the raw type and null type. 
In the work of \cite{Ahmed2019}, they construct benign dataset by  Majestic publisher, which is like the Alexa website.
By constructing the DNS traffic data set, on the one hand, the problem of lacking experimental data to verify the algorithm has been solved. But on the other hand, the realness of DNS tunnel traffic is insufficient. 
Nowadays, network attack is not only confined to the usual hacker but also these planned and purposeful international organizations. 
These special attacks using DNS tunnel technology may be more concealed and difficult to detect. 
Some algorithms validated on general datasets are not effective against such attacks. Also, the high hidden DNS tunnel traffic is often difficult to mark, which make it difficult to construct an effective training set.

\subsubsection{Data Processing}
In the traditional DNS tunnel detection area, there are two main categories: load analysis and flow detection \cite{luo2017}.

DNS load detection mainly uses the agreement of DNS protocol to analyze the data in DNS data package. The principle is that in the DNS data structure, abnormal data is embedded in the normal message structure. Once the DNS message arrives at the attacker's DNS server, the attacker could utlize the embedded data to make a communication. 
Patrick Butler \cite{Butler2011Quantitatively} regards domain names with more than 52 characters in DNS messages as one of the features of DNS tunnel recognition. 
Qi \cite{Cheng2013A} use binary grammar word frequency to detect the frequency of domain name letters in DNS messages. They find that the normal domain names obey Zipf's law while the domain names in DNS tunnels follow a random distribution. 
Bilge \cite{bilge2011exposure} regards detecting the percentage of the longest meaningful substrings in domain names as one of the important features of detecting malicious domain names. 
Skoudis \cite{alshaikhdeeb2015integrating} analyzes the transmission efficiency of the DNS tunnel and find that the abnormal request message would be larger than the normal request message. Based on this, the size of the request message is regarded as the tunnel feature. 
Bilge \cite{bilge2011exposure} et al. analyzes TTL attributes in DNS messages, and the abnormal TTL values of DNS traffic are usually small, which are regarded as characteristics. 

DNS traffic monitoring mainly detects the changes in DNS traffic in the network.
The principle is that when DNS tunnel transmit data, it would generate a large number of DNS messages. This also lead some abnormal phenomena of DNS traffic in the network.
Luo\cite{luo2017} regards the time difference between the first DNS message and the last DNS message as the length of a DNS session and regards the total number of data packets in the DNS session as one of the DNS tunnel characteristics. 

\subsubsection{Detection Model}
Existing detection model mainly include machine learning methods and statistical analysis methods.
Machine learning methods are mainly supervised classification methods.
Ahmed \cite{almusawi2018dns} uses the multi-label SVM method to train the classifier, which improves the classification accuracy.
Wu \cite{zhang2013} proposes an improved stochastic forest decision algorithm to improve the generalization ability of a single classifier.
Saeed \cite{Shafieian2017Detecting} employs an ensemble of machine learning algorithms to detect DNS tunnel. The algorithms used are: Random Forests, K-Nearest and Multi-layer perception.
In the work of \cite{Yassine2018}, they do a survey of DNS tunnel detection techniques, which mainly contains SVM, Naïve Bayes(NB), Decision Tree(DT), and K-nearest Neighbor(KNN) etc.
Lai \cite{Lai2018} proposes a DNS tunnel detector constructed with neural network, which can directly input the first 512 bytes from DNS packet.
Unsupervised machine learning methods that is highly dependent on data characteristics are seldom studied. 
Because these data characteristics are highly uncertain, it is difficult to judge whether or not a DNS tunnel only from several dimensions. 

Statistical analysis is mainly based on the statistical characteristics of DNS traffic messages. There are always differences between normal DNS traffic distribution and abnormal DNS traffic distribution.
By analysing these differences, data are filtered to achieve the target of classification.
Cheng \cite{Cheng2013A} analyzes the distribution of binary characters in domain names and designs a scoring mechanism of binary character frequency. When the domain name score of DNS messages exceeds the threshold, it is determined to be abnormal.
Statistical analysis method highly depends on the construction of the mathematical model, and the quality of model construction is closely related to the judgment accuracy. 
Ellens \cite{Ellens2013Flow} derives 8 traffic-based variables from the IP traffic as an indication of the DNS tunnel, and then uses the threshold method / Brodsky-Darkhovsky method and the distribution-based method to detect the DNS tunnel.

A lot of existing work is based on artificially constructed DNS data, and on this basis, a detection model is built. This will affect the use of DNS tunnel detection in real scenarios. At the same time, machine learning is mainly constrained by the existence of real DNS tunnel data, because real DNS tunnels are not only built by several well-known open-source tunneling tools, but also unknown malware tools. This is also an issue that requires extreme attention and needs to be addressed. Our method relies on DNS traffic in the real environment, and we use statistical rules, which can effectively avoid the being blocked problem of practical applications  due to unreal data.





\section{Methodology Section}

In view of the problem of detecting DNS tunnels at the large enterprises network boundary, we have done a lot of DNS data analysis and found some interesting phenomena. We try to treat it as a starting point and build a theoretical detection model.

\subsection{Phenomena Analysis}
A large enterprise network usually contains several users. The DNS traffic captured in a period of time is mostly generated by these users' online behavior. In addition, there is also some DNS traffic generated by application software. 

Firstly, we analyze the relationship between the length and frequency of domain names over a period of time and find an interesting phenomena that they are similar to the normal distribution curve. 

Is this an accidental phenomenon?

Moreover, we download the top 1 million domain names from the Alexa website and express the length of domain names by the x-axis, and the number of occurrences of domain names by y-axis. After statistical analysis and data fitting, it is found that the length of domain names obeys normal distribution \cite{distribution}, as shown in Figure \ref{fig:top_1m}. 

\begin{figure}[h]
    \centering
    \includegraphics[width=.45\textwidth]{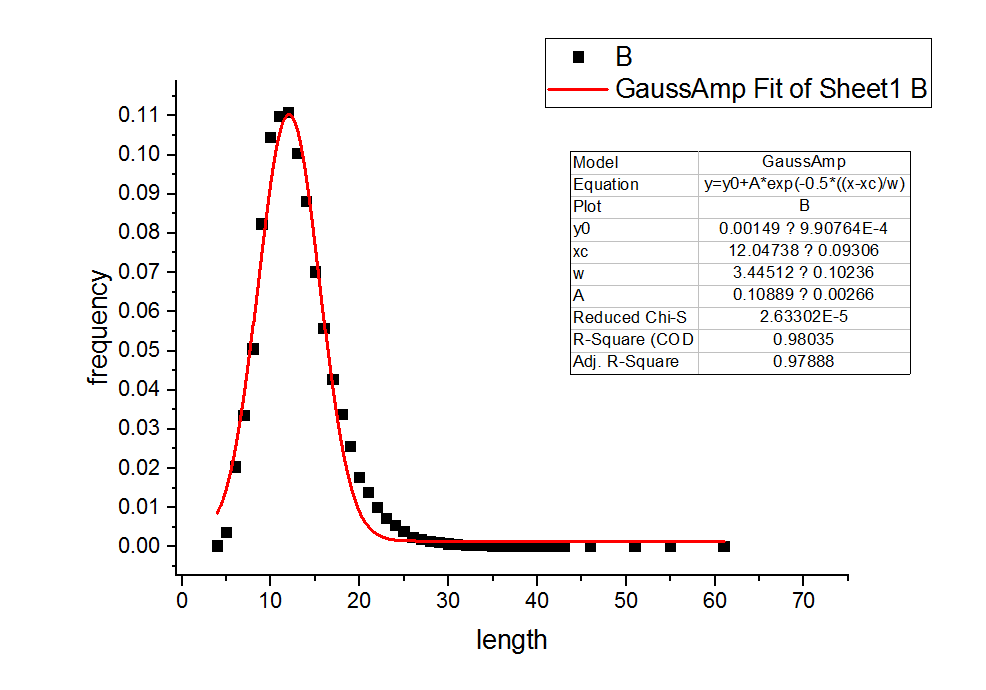}
    \caption{Curve fitting map by Gauss} 
    \textbf{The Figure is fitted by top 1 million domain name length and frequency in Alaxa website using Origin software.}
    \label{fig:top_1m} 
\end{figure}

Figure \ref{fig:top_1m} shows that the length and frequency of domain names fit the normal distribution curve well. In-depth analysis, domain names are applied by companies or users. People tend to use an easy-to-remember, common vocabulary and short spelling as domain names for ease of use. Conversely, very short domain names or very long domain names are always uncommon. 
This at least shows that it is reasonable to fit these data with a normal distribution.
But according to mathematical knowledge, we can know that as long as the parameters are set properly, there will be other curves that can fit these data, such as the GSCA curve in Figure \ref{fig:GCAS2} and polynomial curve in Figure \ref{fig:poly}.
Which curve fitting is used to analyze these data is not particularly important for us, because our focus is on finding some abnormal points. Here we choose to use the normal distribution to analyze these data.

\begin{figure}[h]
    \centering
    \includegraphics[width=.45\textwidth]{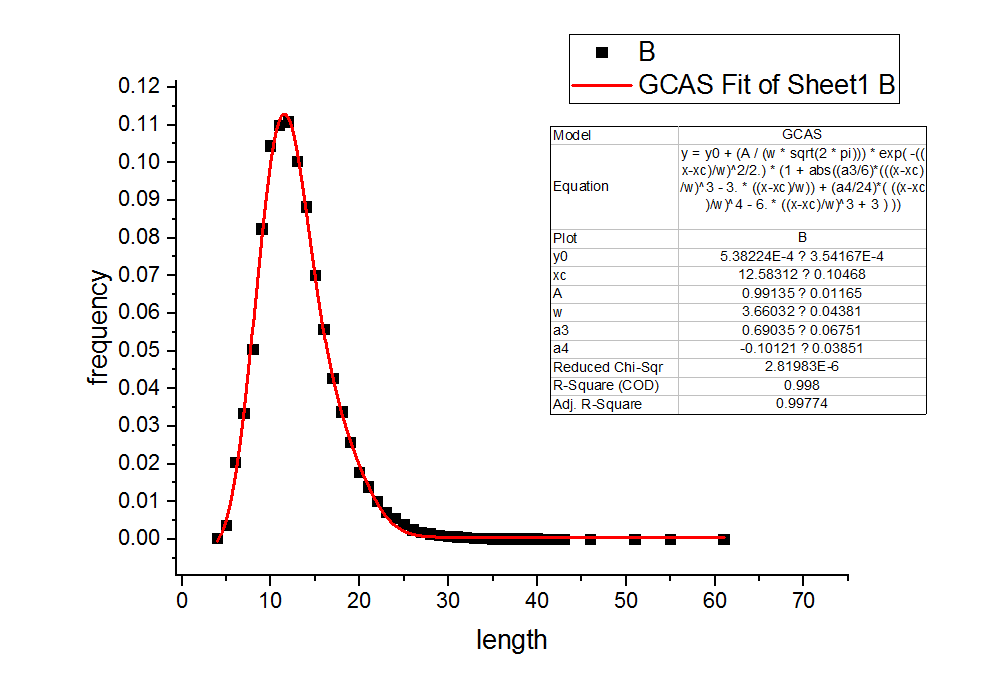}
    \caption{Curve fitting map by GCAS} 
    \textbf{The Figure is fitted by top 1 million domain name length and frequency in Alaxa website using Origin software.}
    \label{fig:GCAS2} 
\end{figure}

\begin{figure}[h]
    \centering
    \includegraphics[width=.45\textwidth]{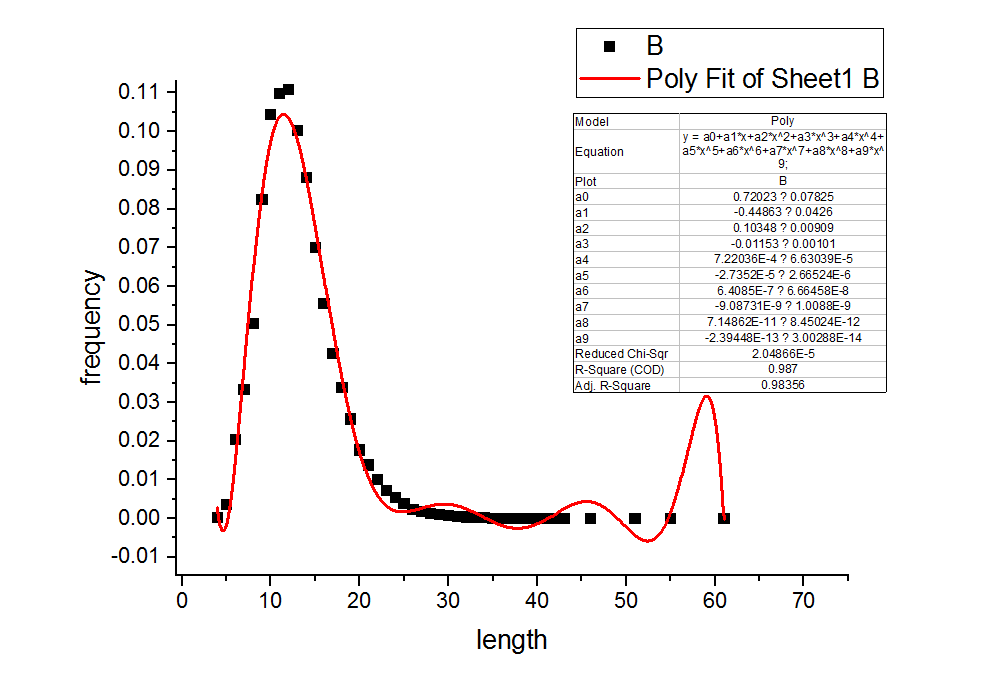}
    \caption{Curve fitting map by Poly} 
    \textbf{The Figure is fitted by top 1 million domain name length and frequency in Alaxa website using Origin software.}
    \label{fig:poly} 
\end{figure}

Secondly, we also found another interesting phenomenon that there will always be abnormal points on the fitting curve of domain name length and frequency, that is, points beyond the reasonable range of fitting curve.

We deeply analyzed the DNS traffic at these abnormal points, and were surprised to find that these abnormal points are usually made due to the existence of a large number of DNS access, in which often contains malicious DNS tunnel traffic. In other words, it is the existence of DNS tunnels that results in a large number of DNS domain name accesses, thus generating anomalies.

\subsection{Raising Thought}
Based on the above phenomena analysis, we may consider that if we can find the fitting curve of traffic flow, then we can screen out its abnormal points through reasonable range. 

And then we can find out the suspected DNS traffic flow through the abnormal points. 

Finally, we can further find the real DNS tunnel traffic through multi-dimensional verification methods.






\subsection{Putting Hypothesis}
If we consider all the sample spaces of DNS domain names as a whole, then a sample space of DNS domain names generated by an internal network can be regarded as a subset of the whole. 
We may make the following assumptions: 

\subsubsection{First} 
We can use a normal distribution to fit the sample space of all DNS domain names.
Note: In this paper, sample space is named domain name sample space(\textbf{DNSS}) that records the same DNS domain name only once. 

\subsubsection{Second} the sample space of DNS domain names generated by a normal internal network also obeys the normal distribution. 

\subsubsection{Third} 
in a normal internal network,  the sample space of domain names visited is certain within a period of time. 
Assuming that each domain name is accessed randomly and averagely 1 to M times, the sample space of DNS domain name constituted by this way would still obeys normal distribution.
Note: In this paper, sample space is referred to as accessing domain name sample space(\textbf{ADNSS}) that count the occurrence number of DNS domain names. 

\subsubsection{Fourth} 
influenced by user preferences, reverse queries, well-known websites and DNS tunnels, abnormal points which is defined as burrs will inevitably show out on the curve under the premise of third hypothesis.
So as long as we can screen out burrs, the burr traffic will inevitably contain abnormal DNS traffic. 

\subsubsection{Fifth} 
further, because the suspicious DNS domain name always appears suddenly, which can be used into the sliding window differential framework. Then, by employing multi-dimension verification methods, the suspicious DNS traffic can be screened out. 

\subsection{Theoretical Proof}

\newtheorem{definition}{define}[section]
\newtheorem{theorem}{theorem}
\newtheorem{proof}{proof}

We first give some definitions.

Define $e$ as the DNS domain name generated by user browser or software application. 

Define $S=\{e\}$ as the sample space of network DNS traffic.

Define a random variable $X=X(e)$ to represent the length of $e$.
\begin{theorem}
Randomly extract a part of $e$ from $S$ to form a new sample space $S^{'}$ where $S'  _{\subseteq} S$, then in the sample space $S^{'}$, $X$ still obeys the same distribution. Assuming that X obeys the distribution in sample space S. 
\end{theorem}

\begin{proof}
If the random variable $X$ is regarded as the value of the random experiment, then the sample space $S$ corresponds to the total $X$ of the observed value of $\left | X \right |$ experiments. 
Then a part of sample $e$ is randomly extracted from the sample space $S$ to form a new sample space. This process can be seen as $M$ times random test of random variable $X$ to obtain $M$ observation values. Then the statistical results of observation values must be consistent with the distribution law of random variable $X$, that is to say, the distribution law of random variable $X$ in the same local space S is consistent, which is testified by this process. That's all. 

\end{proof}

Define function C, whose input is parameter $e$ and output is mapped to a set, as follows:

$$C(e)=\left \{ {e_i}|i\in [1,m],X({e_i})=X(e),1\leq m\leq M  \right \}$$

\begin{theorem}
If there is a sample space $S^"$ which satisfies the equation $${S^"}= \left \{ e'|e'\in C(e),e\in S'  \right \}$$, then in the sample space $S^{"}$, $X$ still obeys the distribution. 
\label{theorem1}
\end{theorem}

\begin{proof}
Define the random variable $Y$ and $Y=\left | C(e) \right |$, which represents the number of set elements. 
According to the definition of $C$, we can get that $Y$ obeys uniform distribution, that is $Y \sim \mho (1,M)$.
Then, 
\begin{equation}
    P(Y=i)=\frac{1}{M},i=1,2,...,M
\label{e1}
\end{equation}
Let us make an assumption that X obeys $N$ distribution, that is $X\sim N$.
Then, we define these two equation.
\begin{equation}
   P(X=i)=P_i,i=1,2,...,N
   \label{e2}
\end{equation}
\begin{equation}
   \left | S' \right |=Q
   \label{e3}
\end{equation}
By analyzing random experiments, it can be considered that random experiments were carried out by $Q$ times, and the values of each random experiment were handled by function C(e), thus constituting the sample space $S^"$.
Let's define a two-dimensional vector $(X,Y)$. 
From the actual experience, the length of the domain name has nothing to do with the number of visits, so we can think that they are independence. 
Therefore $X$ is not related to $Y$. Combined with equation 1, we can get that,
\begin{equation}
    P(X=i,Y=j)=\frac{1}{M}P_i, i=1,2,...,N,j=1,2,...,M
    \label{e4}
\end{equation}
According to the equation \ref{e2} \ref{e3} \ref{e4}, we can find that,
\begin{equation}
\begin{aligned}
        |S"|&=\sum_{i=1}^{N}\sum_{j=1}^{M}j\cdot P(X=i,Y=j)\cdot Q\\
&=\sum_{i=1}^{N}\sum_{j=1}^{M}j\cdot \frac{1}{M}\cdot P_i\cdot Q\\
&=\sum_{i=1}^{N}\frac{M+1}{2}\cdot P_i\cdot Q\\
&=\frac{M+1}{2}\cdot Q
\label{e5}
\end{aligned}
\end{equation}
In the sample space $S^"$, let $|S^"=i|$ denote the frequency where the domain name length is $i$. Then, we can get that,
\begin{equation}
    \left | S"=i \right |=\sum_{j=1}^{M}\frac{1}{M}\cdot P_i\cdot Q=\frac{M+1}{2}\cdot P_i\cdot Q
\end{equation}
\label{e6}
According to the large number theorem, let $f_A$ denote the occurrences number of event A in $n$ times independent repetitive experiments, and let $p$ denote the occurrence probability of event A in each experiment. Then, for any positive number $\varepsilon > 0$, we can get this,
\begin{equation}
    \lim_{n\rightarrow \infty }P\left \{ \left | \frac{f_A}{n}-p \right |<\varepsilon  \right \}=1
    \label{e7}
\end{equation}
That is to say, when $n$ is large enough, the frequency of events can be used to replace the probability of events. 
Also, according to the equation \ref{e5} \ref{e6} \ref{e7}, in the sample space $S^"$, we can get that,
\begin{equation}
    P\left \{ S"=i \right \}=\frac{\left | S"=i \right |}{\left | S" \right |}=P_i=P\left \{ X=i \right \}
    \label{e8}
\end{equation}
Finally, according to the equation \ref{e2} and \ref{e8}, we can find that the sample space $S^"$ also obeys the $N$ distribution.
\end{proof}

In fact, to analyze burrs is to determine a confidence interval and pick out outliers. 
More generally, we give the following definitions.

Define the sample space $E=\{e\}$.

Define a random variables $X$ where $X(e)\in [1,n]$ and $e\in E$. The real meaning of variable $e$ is that it represents the length of sample $e$.

Define a set $K_1$ where $K_1 = \{X(e)|e\in E\}$. The real meaning of variable $K_1$ is that it represents a collection of different length of $e$.

Define a set $K_2$ where $K_2 \subseteq k_1$ and $|K_2| << |K_1|$.

Define a function $G$ where $$G(e)=\{e'_i|i\in [1,g],X(e'_i)\in K_2\}$$, and $g$ is a unknown parameter.

Define a sample space $E'$ where $E'=E\cup G$ and $E\cap G=\O $.

Further, we have known that in the sample space $E$, $X$ obeys the normal distribution law. Then in the sample space $E^{'}$, according to the distribution law of samples, we should study the distribution of $g$ and $K_2$, which corresponds to the distribution of abnormal points. 

Define a function $F$ where $F(x) = P(X<=x)$, which is a theoretical distribution function of random variables $X$ in sample space $E$ and function $f(x)$ is its probability function. 

The function $$S(x)=\frac{k}{n}$$ is defined as the cumulative frequency of the sample ($K$ is the number of samples whose observation value is less than or equal to the number of $x$ and $n$ is the total number of samples), and $s(x)$ is the frequency function. 

Define difference function 
\begin{equation}
    d(x) = s(x) - f(x)
    \label{th3.0}
\end{equation}

By studying the law of $d(x)$, the approximate solution of $K$ can be obtained.  
Here we give the following theorem.
\begin{theorem}
Assuming that under the condition of significance level $\alpha$ and confidence interval $1 - \alpha$, $X$ obeys the normal distribution law and $D$ is the tolerance constant, the upper bound of $D$ satisfies the following formula:
$$d(x) \leq \frac{d}{1+d}\left ( 1 - f(x) \right )$$
\end{theorem}

\begin{proof}
According to K-S test, we define a statistics $D$ where 
\begin{equation}
   D = \underset{x}{max}\left | F(x) - S(x) \right |
   \label{th3.1}
\end{equation}
According to the critical value table of D tested by K-S single sample with the condition of significance level $\alpha$ and confidence interval $1-\alpha$, the value $d$ of $D$ can be determined as tolerance. 
In the sample space $E^{'}$, assuming $|K_2|=1$ satisfied, we can get $g=d\cdot n$, which is the upper bound of $g$. The actual meaning of $g$ is the number of outliers at a burr point. 
If $X$ still obeys the law of normal distribution in sample space $E^{'}$, it must be satisfied with the formula:
\begin{equation}
    s\left ( x \right )\leq \frac{n\cdot f\left ( x \right )+g}{n+g}=\frac{f\left ( x \right )+d}{1+d},
    \label{th3.2}
\end{equation}
supposing that $s(x)$ is always greater than $f(x)$. 

Combined with the formula \ref{th3.2} and \ref{th3.0}, we can get that result,
$$d(x)\leq \frac{f\left ( x \right )+d}{1+d}-f(x)=\frac{d}{1+d}\cdot \left ( 1-f(x) \right )$$
This is evident. 
\end{proof}


\subsection{Hypothesis Validation}

Based on the above theorems, we analyze 5 hypotheses in Section III.C. 

According to the fitting analysis of the top 1 million domain names from Alexa website, as shown in Figure \ref{fig:top_1m},
we find that they fit well on the Normal Distribution curve. 
A reasonable explanation is that people always tend to use domain names of a certain length range. Short and long domain names are both rare, which will make a high in the middle of the curve and a low at both ends. 
Therefore, the length and frequency of DNS show the characteristics similar to Normal Distribution, that is, the middle of the curve is high and the two ends are low. 
Although we haven't found any other proof of this conclusion, based on this publicly trusted DNS database from Alexa website, we consider that the hypothesis 1 is valid. 

According to Theorem 1, if the sample space of all DNS domain names is regarded as a whole, then the sample space of a DNS domain name accessed by an internal network is regarded as a subset of the whole.
It further gets that the distribution of their domain name length values is consistent, which can verify the second hypothesis. 

According to Theorem 2, 
the average visiting number of each domain name is roughly regarded as 1 to M times at random in the sample space of DNS, which exclude famous websites and preferred websites. It gets that the distribution of domain name length value of DNS traffic access space is consistent with the whole, which can verify hypothesis 3. 

Due to the preference of network users, there are always some sites visited more than the average range. According to Theorem 3, a reasonable range can be determined, which is used to screen out the burrs and abnormal flow. This will validate hypothesis 4. 

At this point, the first four hypotheses can be verified. On this basis, according to the fifth hypothesis, the DNS tunnel traffic can be screened out. 

Based on the above analysis, we propose a theoretical detection model to filter DNS traffic. 

\subsection{Expansion Analysis}

As discussed above, this primary and basic assumption is that the DNS sample space follows a normal distribution. If these data show other distribution characteristics or obey other distributions, which can be similar to normal distribution over specified ranges of their parameters, is our method still valid?

We think our approach is still feasible.

We don't think it is very important whether these data obey the normal distribution. In our paper, we believe that it obeys the normal distribution because we use real DNS data in the real world, and then find that it fits well with the normal distribution. If it obeys other distribution in some specific scenarios, this does not affect the use of our method. In this case, it is still possible to use this set of mathematical theories we have given, that is, to establish a reasonable range under this distribution and screen out abnormal points that exceed this interval.
The data contained in these abnormal points is the information that deserves our attention, and thus it can achieve our filter purpose.

\section{Theoretical Detection Model}
According to the above analysis, in order to study DNS traffic, it is necessary to find out the distribution of DNS traffic in the large enterprise network and then screen out abnormal traffic. The whole process is divided into four parts: data fitting, Delineation range, data filtering, and multi-dimensional verification, as shown in the Figure \ref{fig:modeling}.
In the figure \ref{fig:modeling}, Data fitting, Delineation range, and Data filtering are the steps we performed according to the Methodology Section. The fourth step is to use a common method, Multidimensional verification, to determine whether it is a tunnel domain name.

\begin{figure*}[h]
    \centering
    \includegraphics[width=1\textwidth]{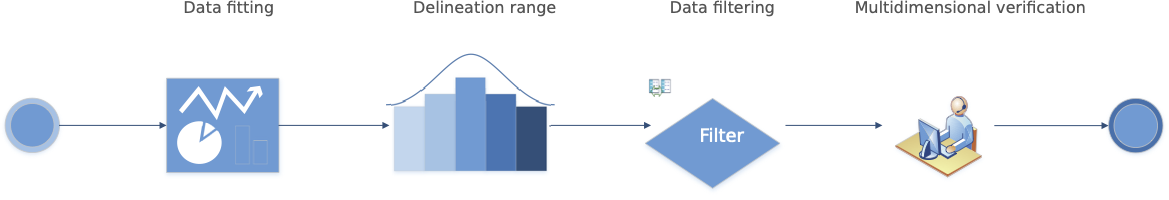}
    \caption{Flow chart of theoretical detection model} 
    \textbf{The flow chart consists of four steps: data filtering, delineation range, data filtering and multidimensional verification.}
    \label{fig:modeling}
\end{figure*}

\subsection{Data Fitting}
we extract all the domain names appearing in the internal network to form a sample space \textbf{NDSS}.
 In \textbf{NDSS}, we represent the length of the domain name as the x-axis and represent the occurrence number of the length of the domain name as the y-axis.
Then, the normal distribution curve s1 of the domain name is constructed through data fitting. 
Furthermore, we handle these DNS to form a sample space \textbf{ANDSS}. In \textbf{ANDSS}, the normal distribution curve s2 of accessing domain names is constructed through data fitting. 

\subsection{Delineation Range }
Due to the influence of network characteristics, DNS traffic in different enterprise networks will be different, which is mainly influenced by some special or user' preference websites.
Therefore, by Theorem 3, we count out the reasonable range of DNS fitting curve. 

\subsection{Data Filtering}
On the curve of sample space \textbf{NDSS} and \textbf{ANDSS}, 
we screen out burr points which are beyond the delimitation range. By these burr points, we can get their related domain name length, which is used to obtain this length DNS traffic.
The white list is used to filter the normal domain names in this DNS traffic, and the rest is suspicious. 

\subsection{Multidimensional Verification}
We take a variety of ways to verify the final result.
By using DNS sandbox technology, suspicious DNS traffic messages are reconstructed and replayed in the sandbox. then we can judge it by observing the response results of DNS server.
By using DNS payload analysis technology, we can further analyze the elements of the DNS package. DNS tunnel usually leads to confusion of query characters, which can be used as a judgment.
By using DNS traffic analysis technology, we can analyze the data interaction in flow. DNS tunnel usually makes massive and sudden traffic, which can be used as a judgment.

\subsection{Extention analysis}
With this model, we can construct the distribution curve of the current data. By filtering the outliers, we can find the length of the domain name that was accessed abnormally. At this domain name length, through further comprehensive verification, we can pick out abnormal domain name information, that is, the suspected DNS tunnel.

\section{Sliding Window Differential Detection Framework}
According to the theoretical detection model proposed in the previous section, we construct a window-based detection framework. By filtering the burrs of each time window, the change of burrs is further compared, and the domain name difference set of the sudden burrs is screened out for judgment. The overall framework is shown in the Figure \ref{fig:framework}.
As can be seen from the figure \ref{fig:framework}, the sliding window method in our Sliding Window Differential Detection Framework is mainly used to cut files, and the differential method is mainly used to filter sudden domain names from adjacent result sets.

\begin{figure*}[h]
    \centering
    \includegraphics[width=1\textwidth]{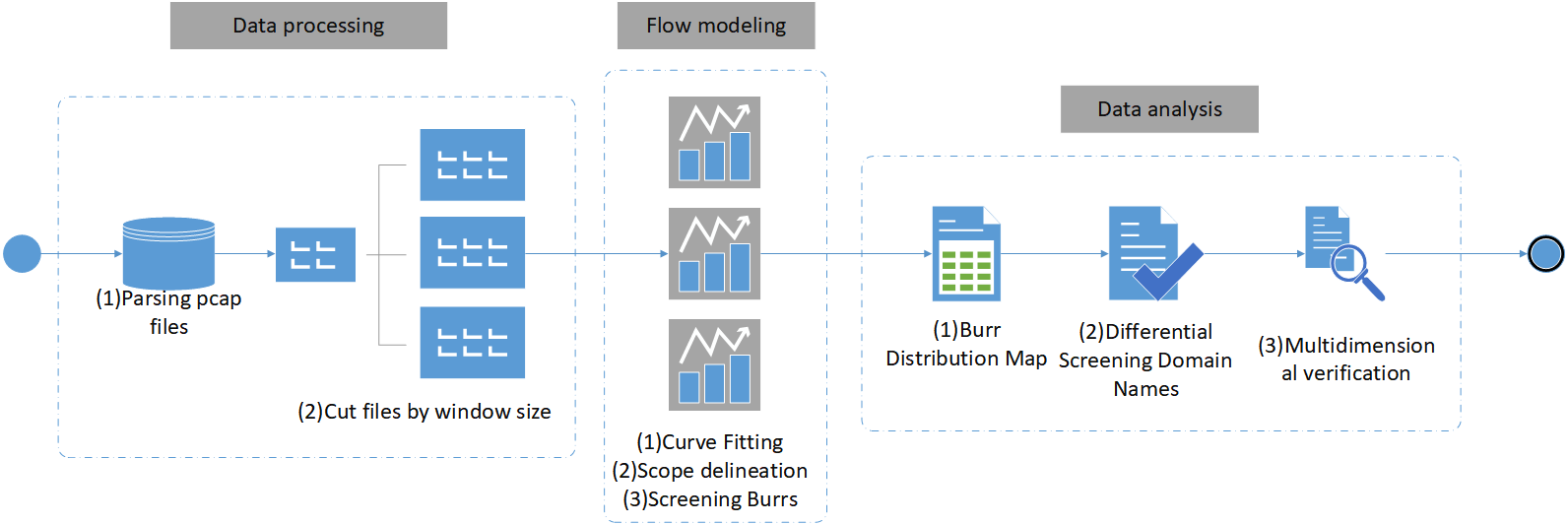}
    \caption{Sliding Window Differential Detection Framework Overview} 
    \textbf{The diagram is composed of 3 steps: data processing, flow modeling and data analysis.}
    \label{fig:framework}
\end{figure*}

\subsection{Data Processing}
In this stage, we first parse the pcap file of traffic and generate CSV file. Then, according to the time window, we cut the file to form the domain name and frequency table with solid time window.

\subsection{Flow Modeling}
In this phase, we use the flow modeling in the previous section to filter the data. After that, the flow in the burr parts is screened out and the burr distribution map is formed. 

\subsection{Data Analysis}
According to the burr distribution map, we select the sudden burr in two adjacent window files, and get the domain names of the burr location in the two window files, then calculate their difference sets. 
Their difference sets reflect the sudden change of domain names in burr areas. Further, we verify the domain names in the difference sets by multidimensional verification and give the final decision results. 

\subsection{Extention Analysis}
The difference between this sliding detection framework and the detection model is that we add sliding and differential methods. In fact, through sliding and differential methods, we can kick out the often reverse queries and hot website visits, and retain the domain name that suddenly appears (such as high-frequency query domain names generated when a DNS tunnel is established), thereby reducing the scope of suspicious domain names. Finally, we use a comprehensive detection method to further exclude domain names that are not DNS tunnels.

\section{Experiment and Evaluation}
\subsection{Experiment}

\subsubsection{Model Testing}

Firstly, we select one month DNS traffic data of an organization in the actual network environment, totaling 645047 DNS query messages. 

According to Theorem \ref{theorem1}, we first extract DNS traffic of \textbf{DNSS}. Its capacity is 48918 domain names. The statistical results of the length and frequency of domain names are shown in figure \ref{fig:fig1}. 

\begin{figure}[h]
    \centering
    \includegraphics[width=.4\textwidth]{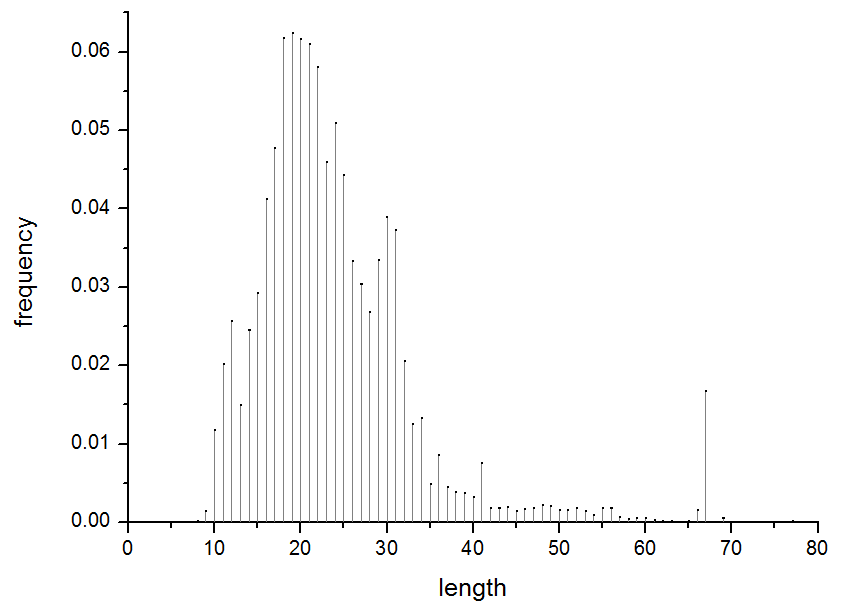}
    \caption{DNSS domain name length and frequency bar chart} 
    \textbf{The figure is formed by extracting the domain name length and frequency of the DNSS.}
    \label{fig:fig1}
\end{figure}

Then, we use the normal distribution to fit it, and the result is shown in figure \ref{fig:fig2}. 

\begin{figure}[h]
    \centering
    \includegraphics[width=.4\textwidth]{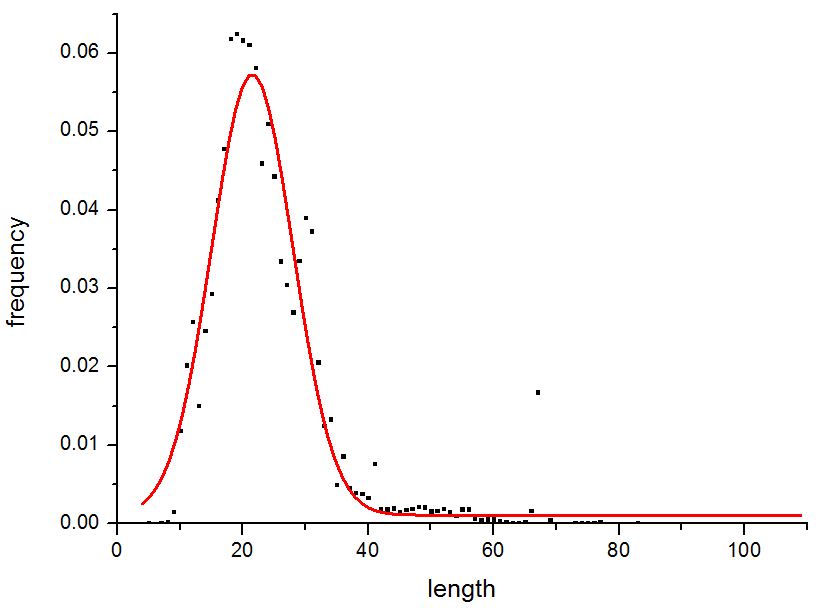}
    \caption{Curve fitting map} 
    \textbf{The Figure is generated by Origin software.}
    \label{fig:fig2}
\end{figure}

In Figure \ref{fig:fig2}, the red solid line is the normal distribution curve generated by data fitting. From this graph, we can find more obvious burr points, which belong to abnormal phenomena. 
According to Theorem 3, the range of $D$ is obtained, as shown in Figure \ref{fig:fig3}. 

\begin{figure}[h]
    \centering
    \includegraphics[width=.4\textwidth]{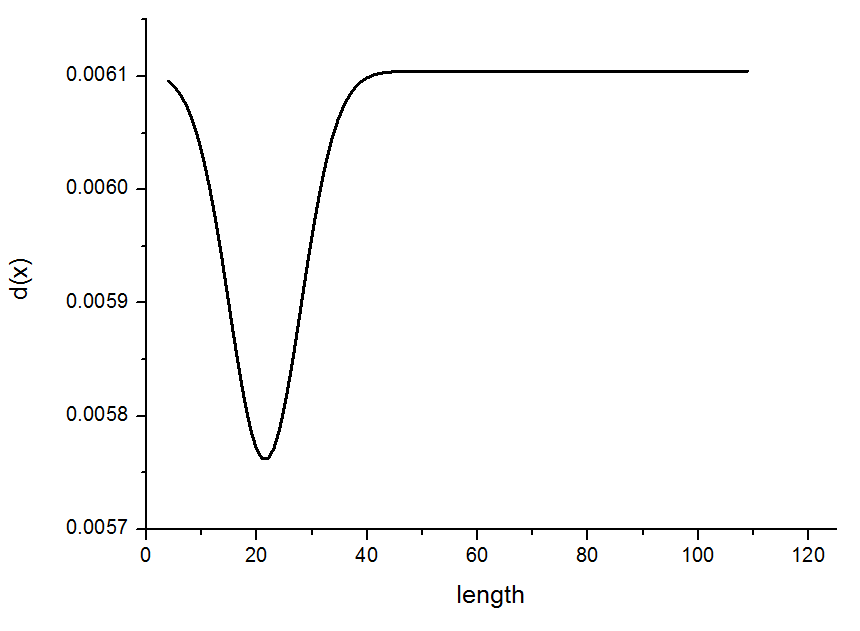}
    \caption{Function curve of d(x)} 
    \textbf{The Figure reflects the change of the function d(x).}
    \label{fig:fig3}
\end{figure}

In Figure \ref{fig:fig3}, we can see that the value of $d(x)$ is affected by $f(x)$. The larger is the value of $f(x)$, the smaller is the value of $d(x)$, and vice versa. It just reflects that the threshold interval is affected by frequency. When the frequency is large, it represents that multiple domain names are of this length. Only when the threshold interval becomes smaller, can it be sensitive to anomalies. When the frequency is small, a small number of domain names are affected by this length, and then the threshold interval will become larger accordingly. 
Further, the delineation range of 95\% confidence interval is obtained, as shown in figure \ref{fig:fig4}. 

\begin{figure}[h]
    \centering
    \includegraphics[width=.4\textwidth]{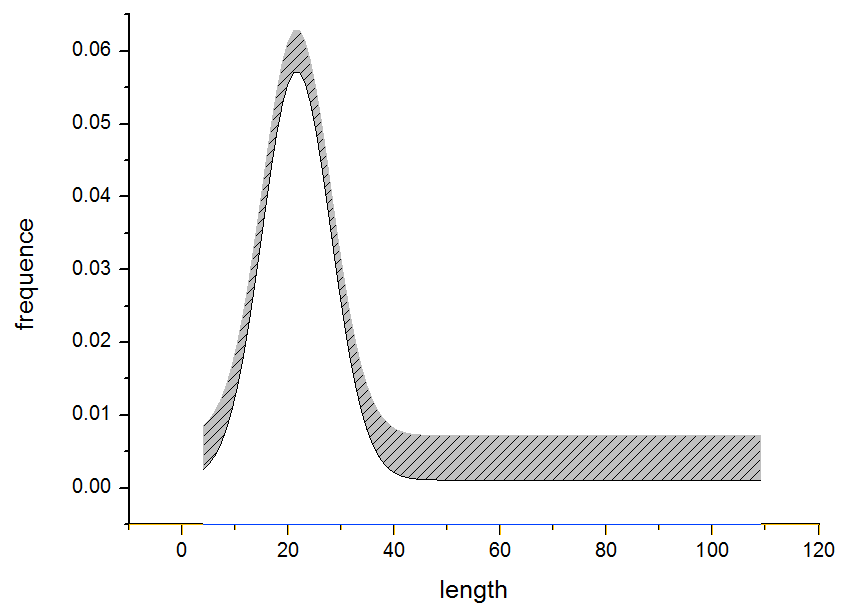}
    \caption{Delineation map for DNSS} 
    \textbf{The shaded portion of the graph reflects the acceptable range of intervals.}
    \label{fig:fig4}
\end{figure}

By comparing with the actual distribution data, we can select the more obvious burr points, that is, the outliers beyond the reasonable range, as shown in figure \ref{fig:fig5}.

\begin{figure}[h]
    \centering
    \includegraphics[width=.4\textwidth]{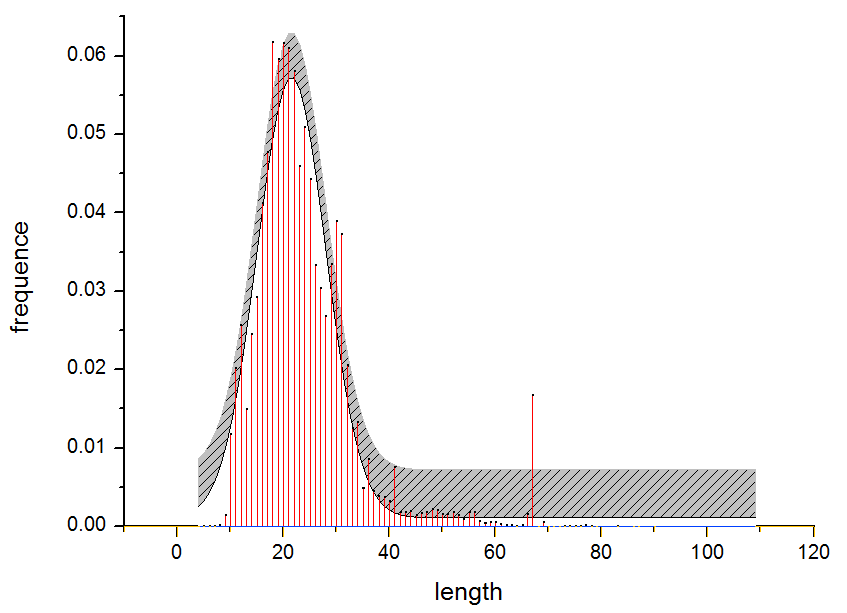}
    \caption{Burr image} 
    \textbf{The red line outside the gray area in the figure is the burr. }
    \label{fig:fig5}
\end{figure}

According to Theorem 2, we extract DNS traffic of \textbf{ADNSS}. Its capacity is 645047 domain names. The statistical results of the length and frequency of domain names are shown in Figure \ref{fig:fig6}. 

\begin{figure}[h]
    \centering
    \includegraphics[width=.4\textwidth]{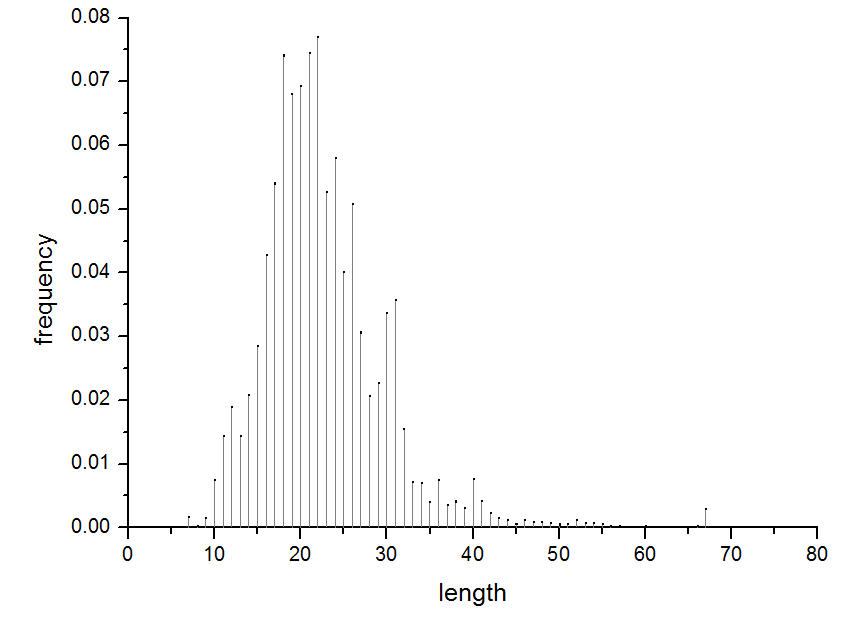}
    \caption{ADNSS domain name length and frequency bar chart} 
    \textbf{The figure is formed by extracting the domain name length and frequency of the ADNSS.}
    \label{fig:fig6}
\end{figure}

Then, the normal distribution is used for fitting, as shown in Figure \ref{fig:fig7}. 

\begin{figure}[h]
    \centering
    \includegraphics[width=.4\textwidth]{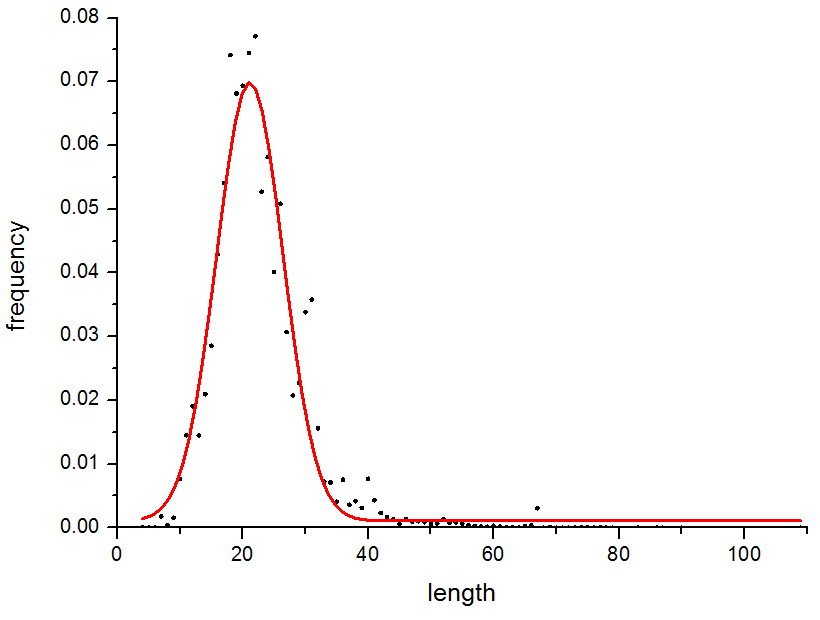}
    \caption{Curve fitting map for ADNSS} 
    \textbf{The Figure is generated by Origin software.}
    \label{fig:fig7}
\end{figure}

The $d(x)$ change of $ADNSS$ is shown in Figure \ref{fig:fig8}. 

\begin{figure}[h]
    \centering
    \includegraphics[width=.4\textwidth]{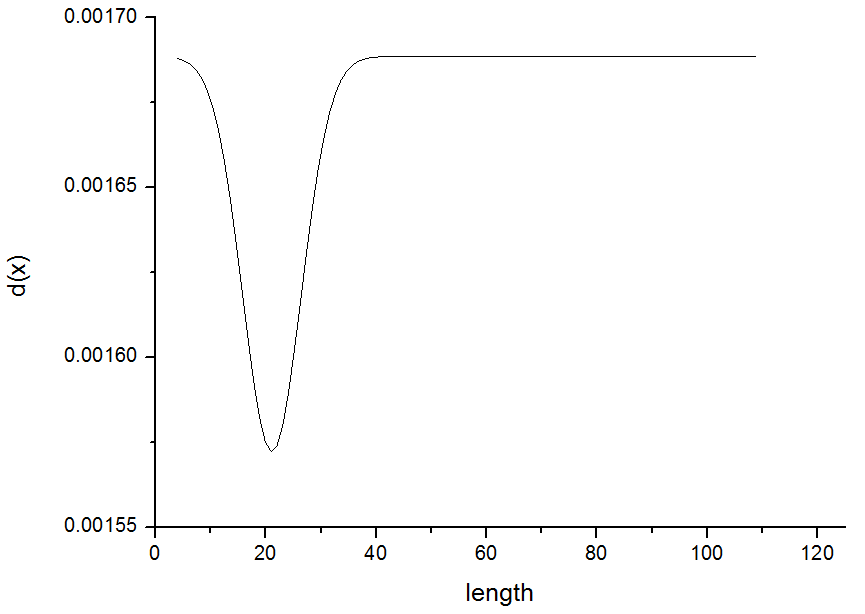}
    \caption{Function curve of d(x) for ADNSS} 
    \textbf{The Figure reflects the change of the function d(x).}
    \label{fig:fig8}
\end{figure}

The delineation range for \textbf{ADNSS} is shown in Figure \ref{fig:fig9}. 

\begin{figure}[h]
    \centering
    \includegraphics[width=.4\textwidth]{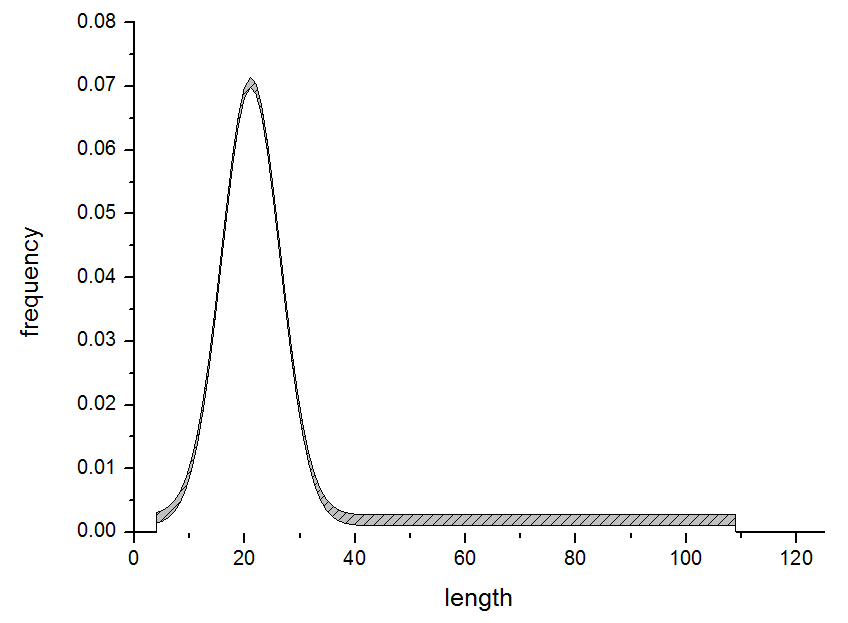}
    \caption{Delineation map for ADNSS} 
    \textbf{The shaded portion of the graph reflects the acceptable range of intervals.}
    \label{fig:fig9}
\end{figure}

The burr distribution for \textbf{ADNSS} is shown in Figure \ref{fig:fig10}. 

\begin{figure}[h]
    \centering
    \includegraphics[width=.4\textwidth]{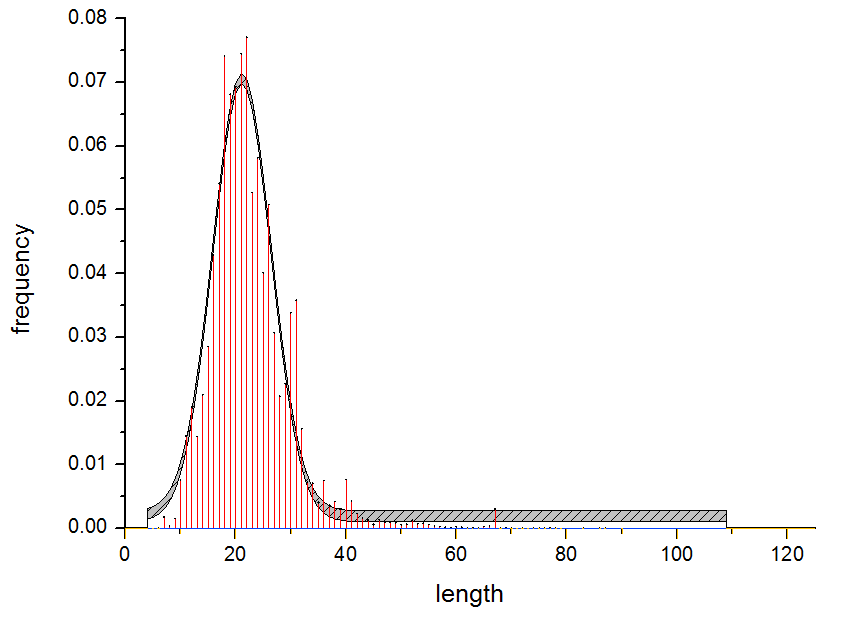}
    \caption{Burr image for ADNSS} 
    \textbf{The red line outside the gray area in the figure is the burr.}
    \label{fig:fig10}
\end{figure}

In combination with Fig \ref{fig:fig5} and Fig \ref{fig:fig10}, we analyze the burr points and find that the burr points are mainly distributed in the location of 12, 18, 21, 22, 30, 31, 40 and 67 length domain name. We further analyze the domain names of burr locations. We find that the number of domain names types in each burr location is not very large, but the number of domain names visiting in each burr location have been accessed a lot, which results in abnormal points. In other words, these domain names have been repeatedly visited many times.We ignore the number of these domain names visiting and only analyze the number of domain names types. And the results are shown in Figure \ref{fig:table1}. 

\begin{figure}[htbp]
    \centering
    \includegraphics[width=.5\textwidth]{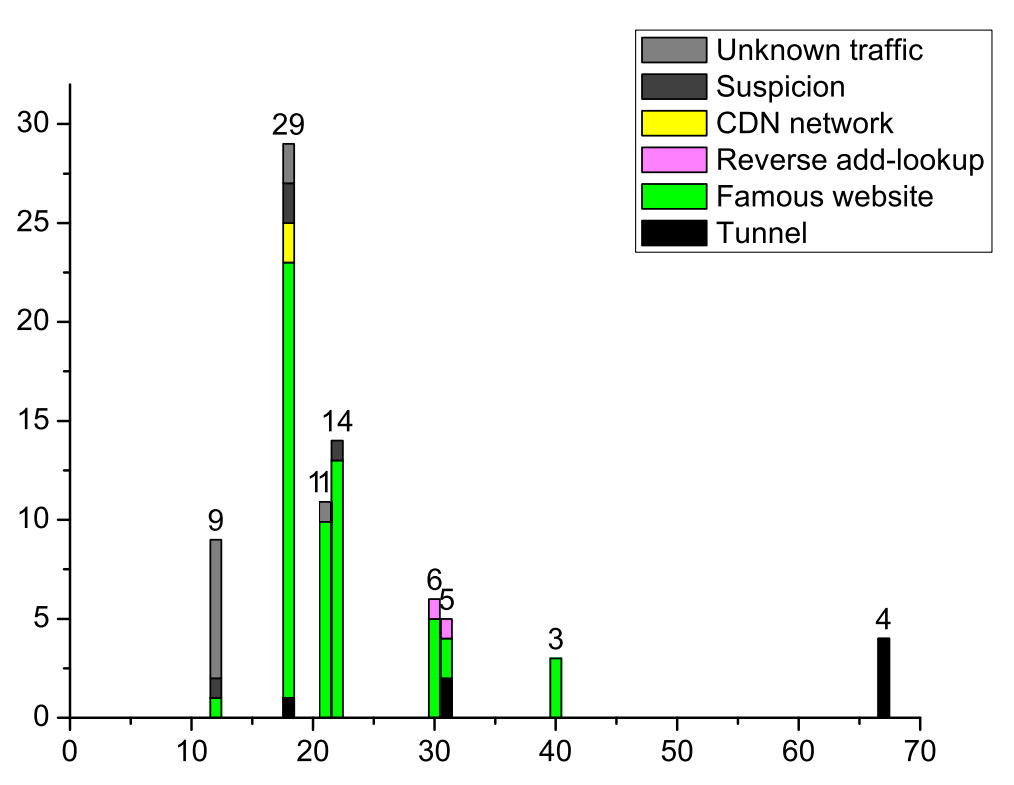}
    \caption{Domain Name Analysis Results for Burr Location} 
    \textbf{The x axis represents the position of the burr, and the y axis represents the number.}
    \label{fig:table1}
\end{figure}

As can be seen from the Figure \ref{fig:table1}, four DNS tunnels are found at Burr 67.\textbf{(Note: Burr 67 in this article refers to abnormal point which is located in domain name of 67 length.)}
At Burr 40, there are a lot of visits to 3 well-known websites. 
At Burr 31, 2 well-known websites and one type of reverse address query are visited extensively, and 2 DNS tunnels are found. 
At Burr 30, there are 5 well-known websites visited extensively, and a large number of reverse address queries. 
At Burr 22, 13 well-known websites are frequently visited, and one malicious web site is frequently visited, presumably due to malicious applications. 
At Burr 21, there are 10 well-known websites and 1 unknown domain names.
At Burr 18, there are 16 legitimate domain names frequently visited, which is presumed to the preference of users.
Else, there are still 6 well-known websites, 2 CDN distribution networks, 2 suspicious domain names, and 1 DNS tunnel. 
At Burr 12, there is 1 well-known website, 7 unknown domain names and 1 malware application. 

In summary, the main reasons for the burr points are well-known websites visiting, reverse address querying, user preference visiting, malware visiting, and DNS tunneling. 

\subsubsection{Window Framework Testing}
When we cut files according to the time window, we do experiments and find that when the time window is less than 7 days, it is difficult to fit, as shown in the figure \ref{fig:fitting1}. 

\begin{figure}[H]
    \centering
    \includegraphics[width=.4\textwidth]{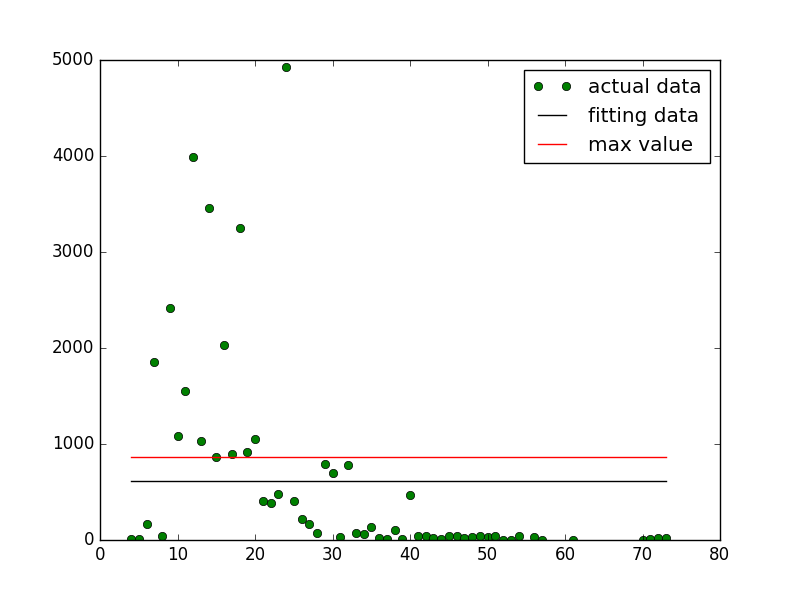}
    \caption{fitting map for less than 7 days datas} 
    \label{fig:fitting1}
\end{figure}

The time window of more than 2 weeks is easier to fit.
The reason for this phenomenon is that our model is based on statistical laws, when the amount of data is relatively small, the statistical laws are not obvious. Only when the amount of data is large, can the statistical law be displayed. 

Here we choose one month as the window size for cutting and fitting, the effect is very good, as shown in the figure \ref{fig:fitting2}. 
\begin{figure}[H]
    \centering
    \includegraphics[width=.4\textwidth]{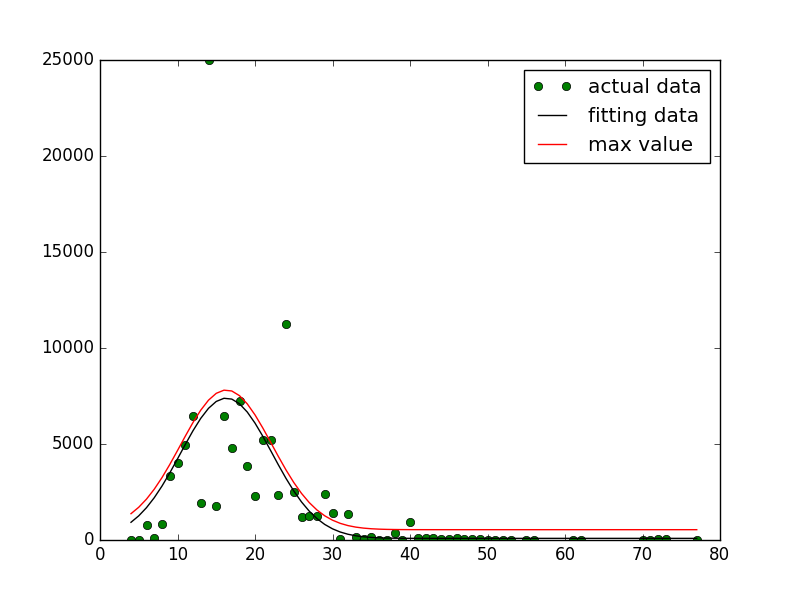}
    \caption{fitting map for one month data} 
    \label{fig:fitting2}
\end{figure}

The generated burr distribution map is shown in the Table 1.

\begin{figure}[htbp]
    \centering
    \textbf{TABLE 1: Burr Distribution Heat Map}
    \includegraphics[width=.4\textwidth]{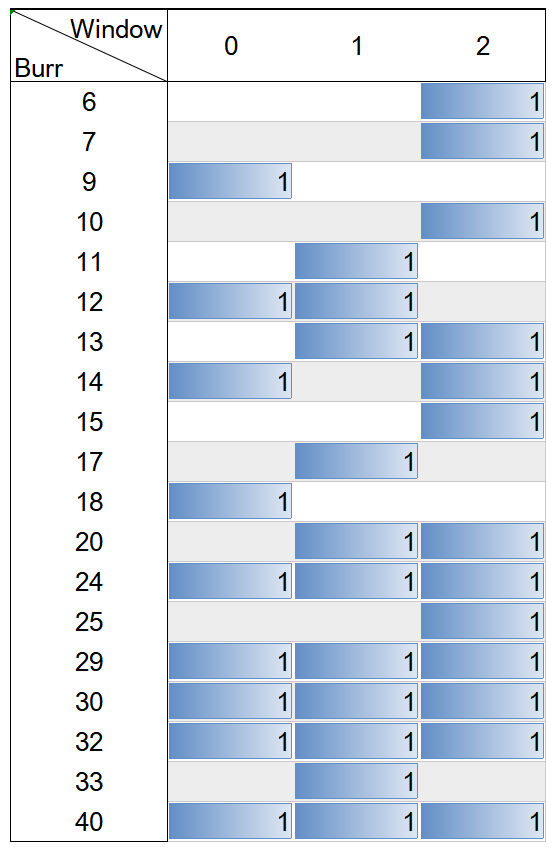}
    \label{fig:burr_distribution}
\end{figure}

In this table 1, the 'Window' of the horizontal represents the window traffic, where 0 represents the traffic of the first month, 1 represents the traffic of the second month, and 2 represents the traffic of the third month. The 'Burr' on the coordinate represents the different positions of the burr.

We select the burr that appears suddenly and extract the domain name of this part. Then we do the difference set to find the domain name that just appear. 
The final results are shown in the Figure \ref{fig:window_sudden_result}. 
\begin{figure}[h]
    \centering
    \includegraphics[width=.5\textwidth]{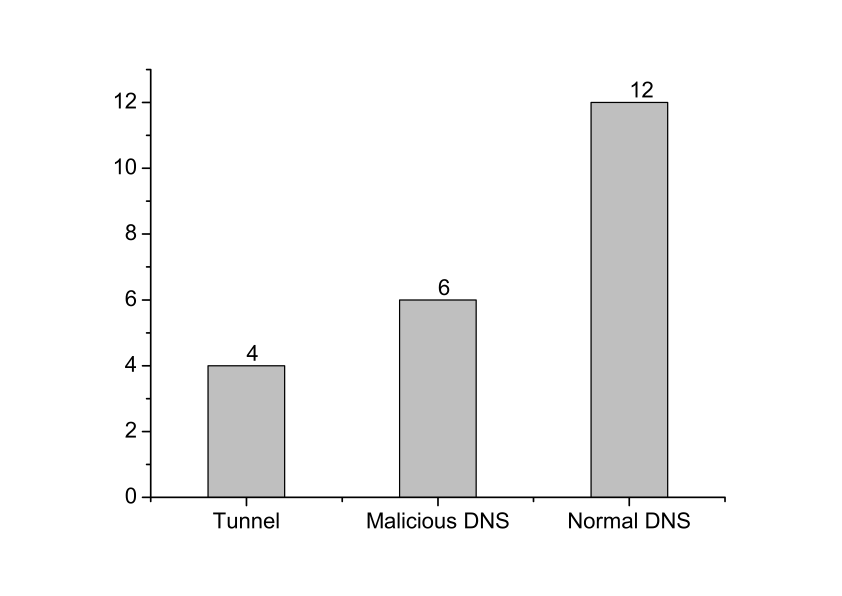}
    \caption{The result of the analysis for the traffic of the sudden burr position} 
    \label{fig:window_sudden_result}
\end{figure}

As can be seen from Figure \ref{fig:window_sudden_result}, 
it greatly improves the screening proportion of suspicious DNS by the window detection framework.
The reason is that normal website access is always continuous,
so that it seldom makes a sudden and massive DNS flow.
But malicious domain name access is often accompanied by a surge in traffic, which is easy to be screened out.

\section{Evaluation}
Our approach is designed to address the issue of how to detect DNS tunnels under non-labeled data conditions. But to verify the effectiveness of our approach, we also conduct a comparative test on the public dataset. We compare it with other methods, both in public data sets and private data sets.
\subsection{Comparison on public datasets}

\subsubsection{data set}

The CTU-13 \cite{CTU} is a public dataset that is captured by the CTU University, Czech Republic. 
We choose the normal traffic in their data set and the tunnel traffic generated by ourselves to build a data set.
Among them, we construct a tunnel environment by the public software Iodine \cite{iodine}, in which we simulate the operations of ping, browser access, and uploading files.

To further differentiate, the tunnel traffic in the training set is generated by employing the 'test.maxsen.tk' domain name. The tunnel traffic in the test set is generated by employing the 'b.tunnel.com' domain name. The tunnel environment for the training set is built on the Internet via VPS. The tunnel environment for the test set is built in a virtual machine environment.

The training set and test set are shown in Figure \ref{fig:public_dataset}.
\begin{figure}[H]
    \centering
    \includegraphics[width=.5\textwidth]{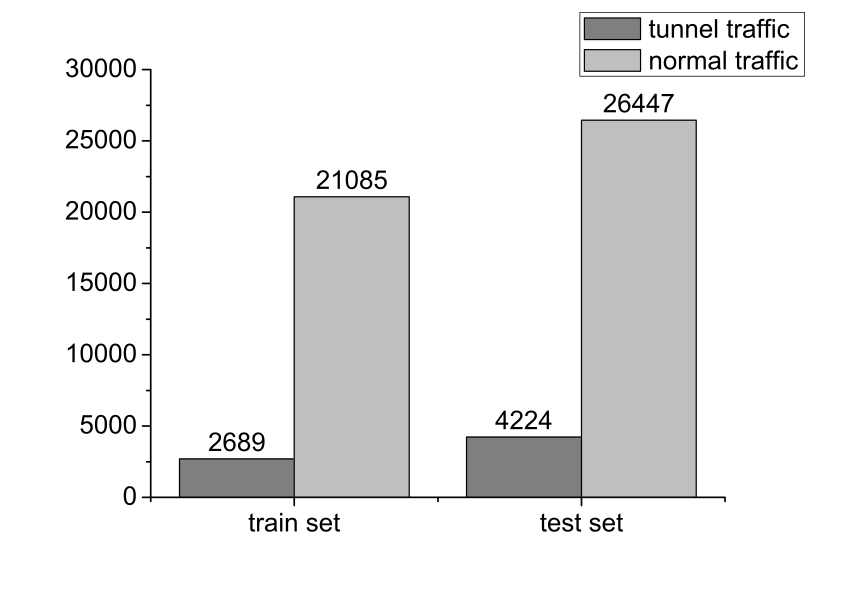}
    \caption{public dataset} 
    \label{fig:public_dataset}
\end{figure}

\subsubsection{Comparison}
We chose the SVM algorithm in the paper \cite{almusawi2018dns} and the neural network algorithm in the paper \cite{homem2017harnessing} as comparison algorithms.
At the same time, we also select IDS \cite{IDS} developed by the Colasoft Co.Ltd as a comparison, which is a rule-based DNS tunnel detection system.

As summarized in the paper of \cite{Yassine2018}, the evaluation indicators used in evaluating detector are different. But most of them choose Accuracy, Precision, Recall, and F1-score as evaluation indicators. We also choose them as our evaluation indicators.
Define TP for true positive, which is original normal traffic and is determined to be normal traffic.
Define FN for false negative, which is original normal traffic and is determined to be tunnel traffic.
Define FP for false positive, which is original tunnel traffic and is determined to be normal traffic.
Define TN for true negative, which is original tunnel traffic and is determined to be tunnel traffic.
Then, 
$$accuracy = \frac{TP + TN}{TP + FN + FP + TN}$$
$$precision = \frac{TP}{TP + FP}$$
$$recall = \frac{TP}{TP + FN}$$
$$F1-score = \frac{2}{{precision}^{-1} + {recall}^{-1}}$$

For SVM and neural network algorithms, we first train them on the training set and then test them on the test set.
Since our method and IDS developed by the Colasoft Co.Ltd do not require a training set, we only test them on the test set.
The result of the burr from our method is shown below in Table 3.

\begin{figure*}[htp]
    \centering
    \textbf{TABLE 3: The domain names of the burr location screened in the test set}
    \includegraphics[width=1\textwidth]{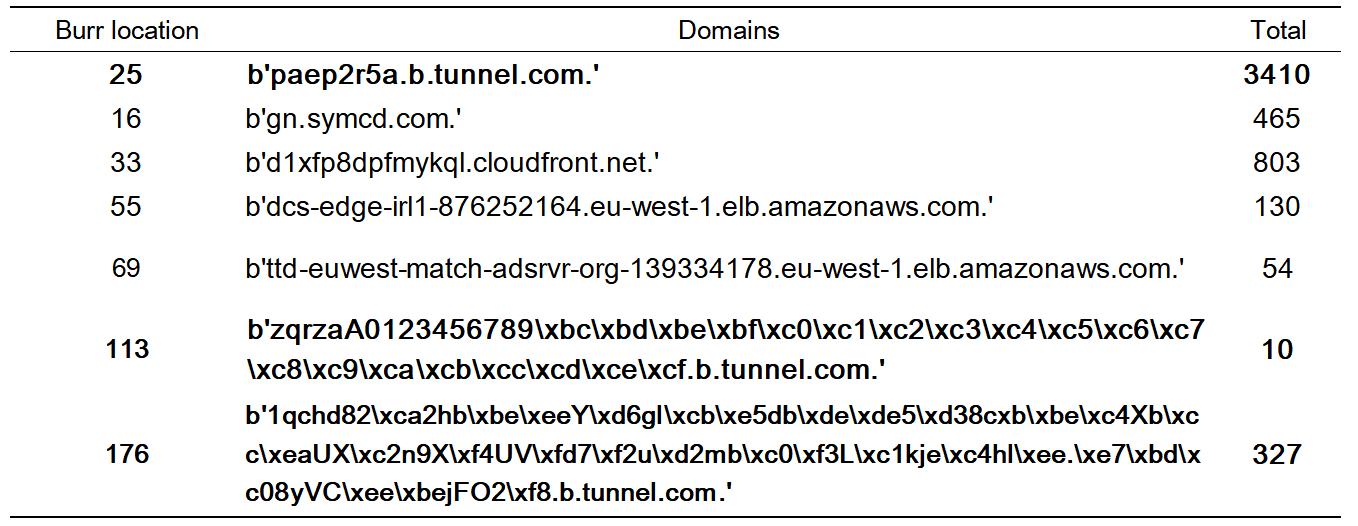}
\end{figure*}
In Table 3, the bold domain name containing the subdomain 'b.tunnel.com' is the tunnel. It can be seen that our method screened out the tunnel flow of 0.97. Moreover, we successfully find the suspicious subdomain 'b.tunnel.com', which has a strong practical significance.

The comparison between our method and the other three methods is shown in Table 4 below.






\begin{figure*}[htp]
    \centering
    \title{}
    \textbf{TABLE 4: Comparison between SVM, Nerual Network, IDS and our method}
    \includegraphics[width=.8\textwidth]{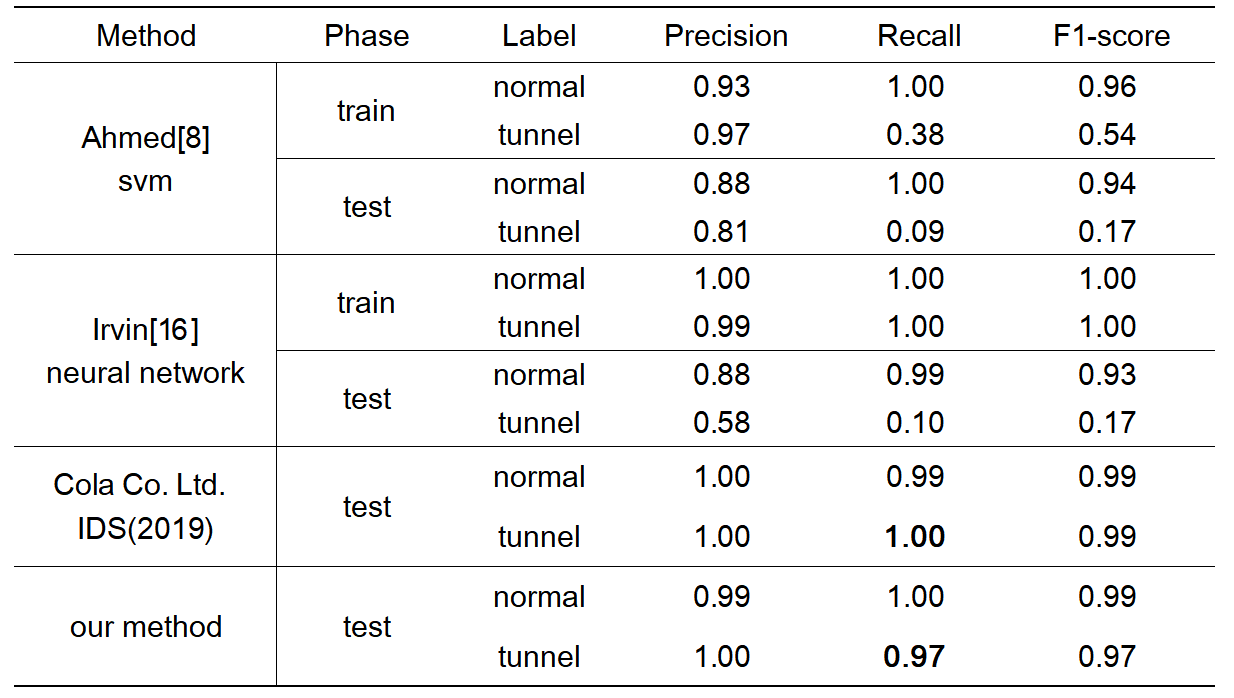}
    
\end{figure*}



In the Table 4, these results are retained two decimal places and we use rounding method. From these results, we can find that for the SVM algorithm on train dataset, its accuracy, precision, recall, and F1-score reach 0.93, 0.22, 0.22, and 0.55, respectively.
For the neural network algorithm on train dataset, its accuracy, precision, recall, and F1-score reach 1.00, 0.99, 1.00, and 1.00, respectively. Their results are very good.

However, the indicators in the test set have declined sharply, especially the recall rate is very low. Among them, the recall rate of the SVM algorithm is 0.09, and the recall rate of the neural network algorithm is 0.10.

In fact, in the original papers of \cite{alshaikhdeeb2015integrating} and \cite{homem2017harnessing}, their results are very good. One possible reason is that 'test.maxsen.tk' is used as tunnel domain name in the training dataset, but 'b.tunnel.com' is used as tunnel domain name in the test dataset.
Their algorithm may only detect the fixed tunnel domain name. But in real traffic, it's impossible for us to know its DNS tunnel form. Once the tunnel domain name changes, it will be difficult to behave well.
It also proves that their algorithms are valid for a particular data set. But once the data changes slightly, the performance of the model drops quickly.

Our method and IDS maintain high classification accuracy and recall. They both have high detection efficiency for known tunnel tools.
Further comparison, IDS reaches 1.00 for the recall rate, which is higher than our 0.97.
Further analysis shows that the reason why the rule-based IDS can achieve a high detection rate is that it has made the rules of tunnel traffic generated by tools such as open source software Iodine.

However, a lot of tunnel traffic in the real environment is not always discovered and made into the rules. And in fact, there are still many unknown tunnels that threaten network security.
So, we are more concerned with the discovery of unknown threats. Therefore, we will validate our method on the private data captured in the actual environment.

\subsection{Comparison on private datasets}

\subsubsection{Data set}
The private data set is ten days of network traffic captured At some internet gateway in a real environment. The total amount of DNS data in the private data is 538064.  
Obviously, these data are unknown and have no tags.
\subsubsection{Comparison}
Because we are analyzing a large number of unknown data, the method of machine learning training classifier is not applicable here. 
Also, we have compared these machine learning methods on the public dataset and analyzed their performance.
Therefore, we choose the IDS(2019 version) developed by the Colasoft Co.Ltd as a comparison to evaluate. 
The comparison results are shown in Figure \ref{fig:compare_our_ids}. 

\begin{figure}[htb]
    \centering
    \includegraphics[width=.5\textwidth]{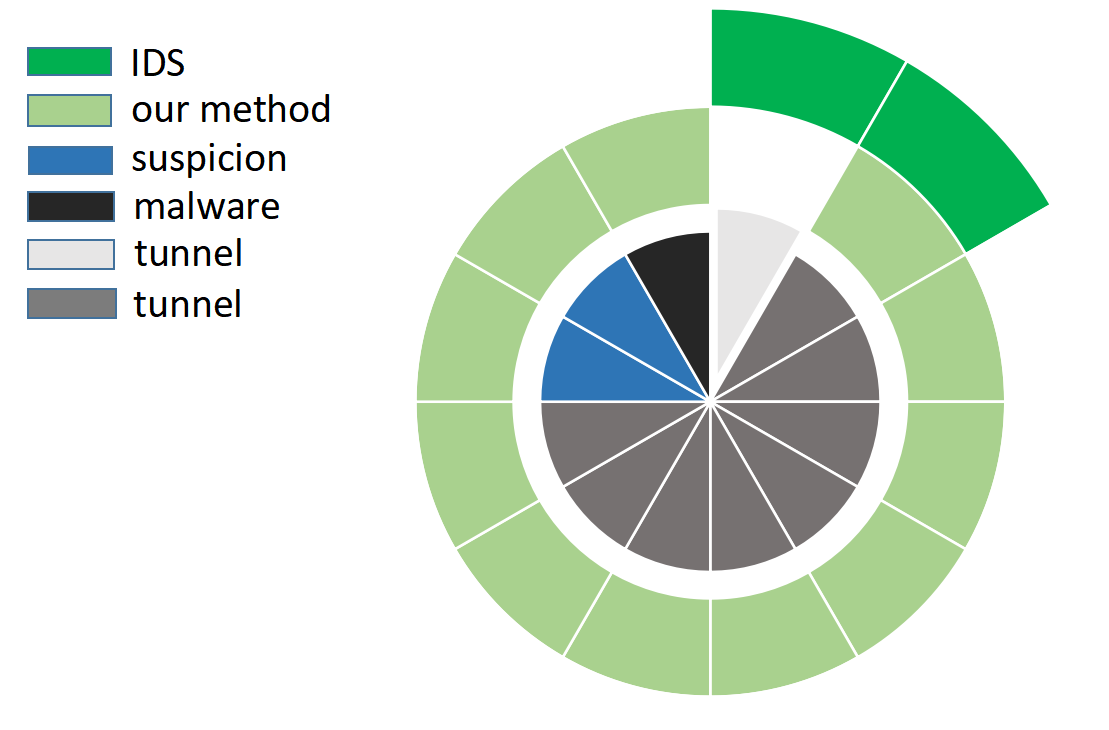}
    \caption{Comparison of our method and IDS on private datasets} 
    \textbf{The picture is composed of three parts: the inner circle, the middle ring, and the outer ring. The inner circle represents all detected traffic. The middle ring represents the traffic detected by our method. The outer ring represents the traffic detected by the IDS.}
    \label{fig:compare_our_ids}
\end{figure}

As can be seen from Figure \ref{fig:compare_our_ids}, 2 DNS tunnels were detected by the IDS. Our method detects 8 DNS tunnels, 2 suspicious DNS domain names and 1 malicious software application domain name. 
It shows that our method has better detection results than IDS.

We further analyze the undetected DNS tunnel. The length of the domain name is 69, and the total number of domain names is 80. 
If the length of the tunnel domain name is not beyond the delimitation range, which can be considered that it is not beyond the tolerance range, there is no burr in our model. 
The tolerance range reflects the acceptable number of domain names. When the number of DNS tunnel messages is small, its bandwidth is very low and data transmission capacity is limited.
In this situation, our method will consider it acceptable. 

In summary, our method has high detection efficiency for unknown threats in the actual environment and has strong practical significance.


\section{Conclusion}
Aiming at the problem that the existing research methods are highly dependent on data sets and not suitable for detecting a large number of unknown DNS tunnel, we study the distribution of DNS traffic and find that the length and frequency of DNS domain names obey the normal distribution. Then some hypotheses are put forward, and three theorems and theoretical proofs are given.       
On this basis, DNS traffic is modeled and a complete theoretical detection framework is constructed. This framework can be used for DNS traffic detection in a real environment and has strong practical significance. 

The limitation of our solution is that it need to get a large amount of DNS data beforehand, so it can not detect DNS tunnels in real time. It can not replace IDS and other DNS tunneling detection methods. But in the field of security, in a view of historical audit and important threat discovery, our method can show good performance even facing a lot of unknown data. This is not satisfied by other methods.

\section{Conflicts of interests}
The authors declare that they have no competing interests.

\section{Data Availability}
Research Data Related to this Submission

Title1: normal traffic dataset

Repository: CTU dataset

url: https://www.stratosphereips.org/datasets-normal

Due to the lack of tunnel traffic data, we will upload the tunnel traffic we generated to GitHub for use by other researchers.

Title2: tunnel traffic dataset

Repository: tunnel dataset

url: https://github.com/maxsen/tunnel\_data

\section{Acknowledgements}
This work is supported by the grants from the National Key Research and Development Program of China(Project No.2018YFB0805000)




\hfil
\hfil
\bibliographystyle{plain}
\bibliography{online}

\begin{thebibliography}{10}

\bibitem{CTU}
dataset website.
\newblock \url{https://www.stratosphereips.org/datasets-normal}, 2019.

\bibitem{DNS}
Dns.
\newblock \url{https://baike.baidu.com/item/dns}, 2019.

\bibitem{IDS}
Ids wiki.
\newblock \url{https://en.wikipedia.org/wiki/Intrusion_detection_system}, 2019.

\bibitem{iodine}
iodine website.
\newblock \url{https://code.kryo.se/iodine/}, 2019.

\bibitem{ks}
Ks test.
\newblock \url{http://www.docin.com/p-47871407.html}, 2019.

\bibitem{distribution}
normal distribution.
\newblock \url{https://baike.baidu.com/item/normal_distribution}, 2019.

\bibitem{Ahmed2019}
Gharakheili H~H Ahmed~J and Raza Q.
\newblock Real-time detection of dns exfiltration and tunneling from enterprise
  networks[c].
\newblock In {\em 2019 IFIP/IEEE Symposium on Integrated Network and Service
  Management (IM).}, IEEE, 2019: 649-653.

\bibitem{almusawi2018dns}
Ahmed Almusawi and Haleh Amintoosi.
\newblock Dns tunneling detection method based on multilabel support vector
  machine.
\newblock {\em Security and Communication Networks}, 2018, 2018.

\bibitem{alshaikhdeeb2015integrating}
Basel Alshaikhdeeb and Kamsuriah Ahmad.
\newblock Integrating correlation clustering and agglomerative hierarchical
  clustering for holistic schema matching.
\newblock {\em Journal of Computer Science}, 11(3):484, 2015.

\bibitem{Berg2019}
Forsberg~D. Berg~A.
\newblock Identifying dns-tunneled traffic with predictive models[j].
\newblock In {\em arXiv preprint}, arXiv:1906.11246, 2019.

\bibitem{bilge2011exposure}
Leyla Bilge, Engin Kirda, Christopher Kruegel, and Marco Balduzzi.
\newblock Exposure: Finding malicious domains using passive dns analysis.
\newblock In {\em Ndss}, pages 1--17, 2011.

\bibitem{Born2010Detecting}
Kenton Born and David Gustafson.
\newblock Detecting dns tunnels using character frequency analysis.
\newblock {\em Corr}, 2010.

\bibitem{Born2010NgViz}
Kenton Born, David Gustafson, Kenton Born, and David Gustafson.
\newblock Ngviz: Detecting dns tunnels through n-gram visualization and
  quantitative analysis.
\newblock {\em Computer Science}, page~47, 2010.

\bibitem{Butler2011Quantitatively}
Patrick Butler, Kui Xu, and Danfeng Yao.
\newblock Quantitatively analyzing stealthy communication channels.
\newblock In {\em International Conference on Applied Cryptography Network
  Security}, 2011.

\bibitem{Cheng2013A}
Qi~Cheng, Xiaojun Chen, Xu~Cui, Jinqiao Shi, and Peipeng Liu.
\newblock A bigram based real time dns tunnel detection approach.
\newblock {\em Procedia Computer Science}, 17:852--860, 2013.

\bibitem{Ellens2013Flow}
Wendy Ellens, Piotr Żuraniewski, Anna Sperotto, Harm Schotanus, and Erik
  Meeuwissen.
\newblock Flow-based detection of dns tunnels.
\newblock In {\em Proceedings of the 7th IFIP WG 6.6 international conference
  on Autonomous Infrastructure, Management, and Security: emerging management
  mechanisms for the future internet - Volume 7943}, 2013.

\bibitem{homem2017harnessing}
Irvin Homem and Panagiotis Papapetrou.
\newblock Harnessing predictive models for assisting network forensic
  investigations of dns tunnels.
\newblock 2017.

\bibitem{Homem2018Information}
Irvin Homem, Panagiotis Papapetrou, and Spyridon Dosis.
\newblock Information-entropy-based dns tunnel prediction.
\newblock In {\em IFIP International Conference on Digital Forensics}, 2018.

\bibitem{Lai2018}
Huang B~C Lai C~M and Huang~S Y.
\newblock Detection of dns tunneling by feature-free mechanism[c].
\newblock {\em IEEE Conference on Dependable and Secure Computing (DSC). IEEE},
  2018: 1-2.

\bibitem{Shafieian2017Detecting}
Saeed Shafieian, Daniel Smith, and Mohammad Zulkernine.
\newblock Detecting dns tunneling using ensemble learning.
\newblock In {\em International Conference on Network System Security}, 2017.

\bibitem{Yassine2018}
Khalife~J Yassine~S and Chamoun M.
\newblock A survey of dns tunnelling detection techniques using machine
  learning[c].
\newblock {\em BDCSIntell}, 2018: 63-66.

\bibitem{luo2017}
youQ Luo, Shengli Liu, Meng Yan, and DongYing Wu.
\newblock Dns tunnel trojan detection method based on communication behavior
  analysis.
\newblock {\em Journal of Zhejiang University (Engineering Science)},
  51(9):1780--1787, 2017.

\bibitem{zhang2013}
Siyu ZHANG, Futai Zou, Lu~hua WANG, and Ming CHEN.
\newblock Detecting dns-based covert channel on live traffic.
\newblock {\em Journal on Communications}, (5):143--151, 2013.

\end{thebibliography}

\end{document}